\documentclass[twoside,11pt]{article}
\usepackage{amssymb,amsmath,amsthm,amsfonts, epsfig}           
\usepackage{mathtools} 
\usepackage{hyperref}
\usepackage{eucal}
\usepackage{a4wide}
\usepackage[svgnames]{xcolor}
\usepackage{graphicx}
\usepackage{tikz-cd}
 \usepackage{fancyhdr}
 \usepackage{MnSymbol}
\usepackage{bm}
\usepackage{subcaption}
\usepackage{caption}
\usepackage{float}

\title{Some remarks about the centre of mass of two particles in spaces of constant curvature}

\author{  Luis C.~Garc\'ia-Naranjo }

\numberwithin{equation}{section}
\numberwithin{table}{section}
\numberwithin{figure}{section}

\hypersetup{colorlinks=true,linkcolor=violet,citecolor=purple}

\newtheorem{theorem}{Theorem}[section]

\newtheorem{proposition}[theorem]{Proposition}

{\theoremstyle{definition}

\newtheorem{remark}[theorem]{Remark}
\newtheorem*{remarks*}{Remarks}

}

\newtheorem{innercustomgeneric}{\customgenericname}
\providecommand{\customgenericname}{}
\newcommand{\newcustomtheorem}[2]{%
  \newenvironment{#1}[1]
  {%
   \renewcommand\customgenericname{#2}%
   \renewcommand\theinnercustomgeneric{##1}%
   \innercustomgeneric
  }
  {\endinnercustomgeneric}
}

\newcustomtheorem{customhyp}{}

\def\headcolour{\color{Grey}}
\def\restr#1{\,\vrule height1.2ex width.4pt
  depth0.8ex\lower0.4ex\hbox{\scriptsize $\,#1$}}

\newcommand{\R}{\mathbb{R}}

\newcommand{\g}{\mathfrak{g}}
\newcommand{\so}{\mathfrak{so}}

\newcommand{\SO}{\mathrm{SO}}

\newcommand{\q}{\bm{q}}
\newcommand{\vv}{\bm{v}}
\def\d{\mathrm{d}}

\begin{document}

\maketitle

\abstract{The  concept of centre of mass of two particles in 2D spaces of constant Gaussian curvature is discussed by recalling the 
notion of  ``relativistic rule of lever" introduced by Galperin [Comm. Math. Phys. {\bf 154} (1993), 63--84] and comparing it with two other definitions of centre of mass that arise naturally on the treatment of the 
2-body problem in spaces of constant curvature: firstly  as the collision point of particles that are initially at rest, and secondly as the centre of rotation of steady rotation solutions. It is shown
that if  the particles
have distinct masses then these definitions are equivalent only if the curvature vanishes and instead lead to three different notions of centre of mass in the general case.}

\vspace{0.3cm}
\begin{flushright}\em  Dedicated to James Montaldi. \end{flushright}

\section{Introduction}

Consider two particles with masses $\mu_1, \mu_2>0$ located at $\q_1, \q_2 \in \R^2$. As is well-known, their centre of mass is the point 
\begin{equation}
\label{eq:c-of-mass}
\bar \q:=\frac{\mu_1\q_1+ \mu_2\q_2}{\mu_1+\mu_2} \in \R^2.
\end{equation}
Denote by $\ell$ the  line segment connecting $\q_1$ and $\q_2$. Then the centre of mass 
mass lies on $\ell$ and satisfies
\begin{equation}
\label{eq:levRule0}
\mu_1 r_1 =\mu_2r_2,
\end{equation}
where $r_j =|\bar \q -\q_j|$ is the Euclidean distance between $\bar \q$ and $\q_j$, $j=1,2$.  The familiar Eq.~\eqref{eq:levRule0}  may be derived from the  following three different characterisations
of the centre of mass:
\begin{description}
\item[C1.  Lever rule:]  suppose that the segment $\ell$ is a massless horizontal beam. The centre of mass $\bar \q$ is the unique point on the beam such that, if a hinge is located at this point, then the 
torques exerted by the two masses balance.

\item[C2. Collision point:] suppose that the particles are under the influence of  an attractive potential  force, depending  only  on their mutual distance  (for instance gravity).  If the particles are initially at rest 
then they will eventually collide at the centre of mass $\bar \q$.

\item[C3. Centre of steady rotation:]  suppose again that the particles are under the influence of  an attractive potential  force, depending  only  on their mutual distance   (for instance gravity).
If a solution exists  in which the particles rotate uniformly along concentric circles, maintaining a constant distance at all time, then the centre of mass $\bar \q$ coincides with the centre of the circles.

\end{description}

\begin{remark}
 It is well-known that solutions to the 2-body problem satisfying the conditions in C3 do exist
 for any positive distance between the particles and for any attractive potential.
 In fact they are the starting point to consider the
planar circular restricted 3-body problem.
\end{remark}
 
 These characterisations are illustrated in Figure\,\ref{F1}.

\begin{figure}[h]
\captionsetup{width=0.8\textwidth,font=small}
\centering
  \begin{subfigure}[t]{0.27\textwidth}
  \centering
  \vspace{-2.1cm}
    \includegraphics[width=0.8\textwidth]{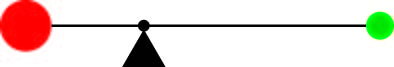}
    \put(-63,15){\small $\bar \q$}
        \put(-106,15){\small $\mu_1$}
         \put(-78,2){\small $r_1$}
           \put(-35,2){\small $r_2$}
           \put(0,15){\small $\mu_2$}
      \vspace{1.4cm}
    \caption{C1. Lever rule.}
    \label{fig:C1}
  \end{subfigure}
  \begin{subfigure}[t]{0.27\textwidth}
    \vspace{-2.2cm}
  \centering
    \includegraphics[width=0.8\textwidth]{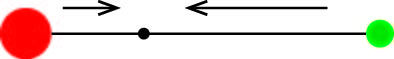}
      \put(-63,11){\small $\bar \q$}
        \put(-104,13){\small $\mu_1$}
         \put(-78,-2){\small $r_1$}
           \put(-35,-2){\small $r_2$}
           \put(-2,13){\small $\mu_2$}
     \vspace{1.5cm}
    \caption{C2. Collision point.}
    \label{fig:2}
  \end{subfigure}
  \begin{subfigure}[t]{0.30\textwidth}
  \centering
    \includegraphics[width=0.8\textwidth]{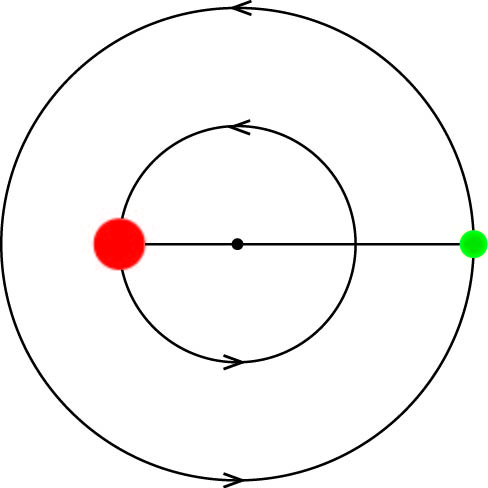}
     \put(-57,58){\small $\bar \q$}
        \put(-95,58){\small $\mu_1$}
         \put(-70,46){\small $r_1$}
           \put(-25,46){\small $r_2$}
           \put(0,58){\small $\mu_2$}
    \caption{\centering C3. Centre of steady \newline rotation.}
    \label{fig:2}
  \end{subfigure}
  \caption{ Illustration of the centre of mass according to the characterisations C1, C2 and C3.}
  \label{F1}
\end{figure}
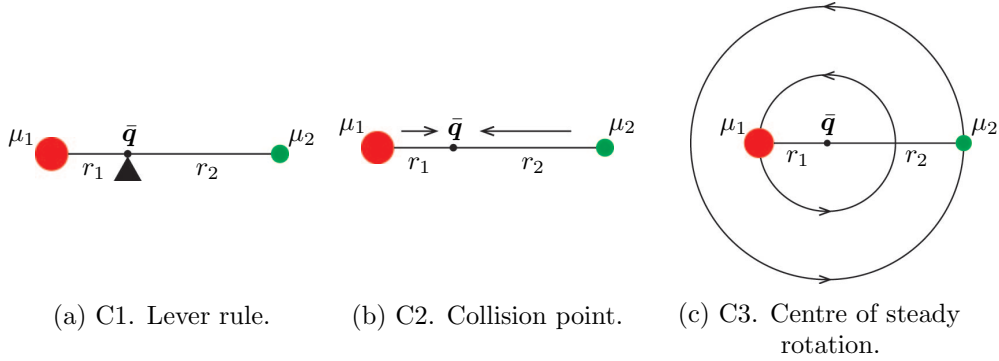

In this note we consider the generalisation of the concept  of  centre of mass of two particles to
2D  spaces of constant non-zero Gaussian curvature $\kappa$. For $\kappa>0$ this is the sphere of radius $1/\sqrt{\kappa}$, and for $\kappa<0$ there are various well-known models and we choose to work 
on the hyperboloid or Lorenz model (see 
section~\ref{S:basic defns} for details). 

We consider three different  definitions of the centre of mass obtained by enforcing conditions C1, C2 and C3 in spaces of constant curvature. As is natural to expect, independently of the choice of condition, the resulting centre of mass lies along the  shortest geodesic connecting the two masses, and  coincides with the mid-point on this geodesic if the masses are equal. To guarantee uniqueness,
for $\kappa>0$  we do not consider antipodal configurations on the sphere. 

 Throughout the paper we will denote by $r_j$, $j=1,2$, the Riemannian distance from $\mu_j$ to the centre of mass and by $r=r_1+r_2$ the total Riemannian
distance between the masses. In accordance to what was said before,
 if the masses are equal then $r_1=r_2$. On the other hand, if
 they are distinct, then a certain generalisation of Eq.~\eqref{eq:levRule0} holds.
Interestingly, the form of this generalisation depends on which characterisation of the centre of mass, C1, C2 or C3, one is dealing with, as we now explain.

The generalisation of the rule of lever C1 was considered in depth by Galperin~\cite{Galperin}. 
His work goes beyond the definition of the centre of mass of two particles and 
succeeds to define the {\em centroid}\footnote{the centroid is located at the centre of mass but also  has
a mass assigned to it, that in the euclidean case  is the total mass of the particles.} 
  of $N$ particles in spaces of constant curvature of any dimension. 
  Galperin's  paper only deals with the values of the curvature $\kappa=\pm1$,  but his construction may be naturally  extended to 
  all values of $\kappa\in \R$ (see section~\ref{S:Galperin}).  In this case Eq.~\eqref{eq:levRule0} is replaced  by {\em the relativistic rule of lever}: 
  \begin{equation}
  \label{eq:GC1}
\mbox{Generalisation of C1} \implies \begin{cases} \mu_1\sin (\sqrt{\kappa} r_1) =\mu_2 \sin (\sqrt{\kappa} r_2), \quad \mbox{if $\kappa>0$}; \\  \mu_1\sinh (\sqrt{-\kappa} r_1)  =\mu_2 \sinh (\sqrt{-\kappa} r_2), \quad \mbox{if $\kappa<0$,}
\end{cases}
\end{equation}
(see Eqs. (5) and (6) in~\cite{Galperin})  which determines the centre of mass uniquely.

To the best of my knowledge, the generalisation of C2 has not been considered before apart from a private conversation that I had with James Montaldi while 
visiting him in Manchester in 2016.  In this paper we develop on this conversation and 
prove that this generalisation recovers the functional form of Eq.~\eqref{eq:levRule0}, namely,
  \begin{equation}
  \label{eq:GC2}
\mbox{Generalisation of C2} \implies \mu_1r_1=\mu_2r_2 \quad \mbox{for all $\kappa \in \R$.}
\end{equation}

Finally,  the generalisation of C3 appears  in recent works concerned with the classification of relative equilibria of the 2-body problem in spaces of constant curvature
\cite{BMK, GNMPR, BGNMM, GNM} and Eq.~\eqref{eq:levRule0} is replaced by
 \begin{equation}
  \label{eq:GC3}
\mbox{Generalisation of C3} \implies \begin{cases} \mu_1\sin (2\sqrt{\kappa} r_1) =\mu_2 \sin (2\sqrt{\kappa} r_2), \quad \mbox{if $\kappa>0$}; \\  \mu_1\sinh (2\sqrt{-\kappa} r_1)  =\mu_2 \sinh (2\sqrt{-\kappa} r_2), \quad \mbox{if $\kappa<0$.}
\end{cases}
\end{equation}
The relations in \eqref{eq:GC3}  uniquely specify the centre of mass except when  $\kappa>0$ and $r= \frac{\pi}{2\sqrt{\kappa}}$ (i.e. when the masses subtend a right angle). 
In this exceptional case, the centre of mass is undefined for distinct masses and we define it  to coincide with 
the midpoint if  the masses are equal.

A comparison between these three generalisations is illustrated in Figure~\ref{F2} below.
 There it  is assumed that $\mu_2= 2\mu_1$, $r_1=1$, and the 
 value of $r_2$ is graphed as a function of the curvature $\kappa$ according
to \eqref{eq:GC1},  \eqref{eq:GC2} and  \eqref{eq:GC3}. The resulting relation
between $\kappa$ and $r_2$ is one-to-one except for  Eq. \eqref{eq:GC3} when $\kappa>0$.
In this case  $r_2$ is  a two-valued function of $\kappa$: one of the possible values of $r_2$ satisfies 
$0<r_2+1< \pi/ 2\sqrt{\kappa}$, 
corresponding to an  acute arc between the masses, and the other value of $r_2$ satisfies 
$ \pi/ 2\sqrt{\kappa}<r_2+1<\pi/\sqrt{\kappa}$, corresponding to an obtuse  arc.
 These two values respectively correspond to the acute and obtuse relative equilibria determined recently 
in \cite{BGNMM, GNM}. Note that the three graphs intersect only when $\kappa=0$ and instead
lead to  different notions of centre of mass for  $\kappa\neq 0$.

 \begin{figure}[h]
 \captionsetup{width=0.8\textwidth,font=small}
\centering
\includegraphics[width=9cm]{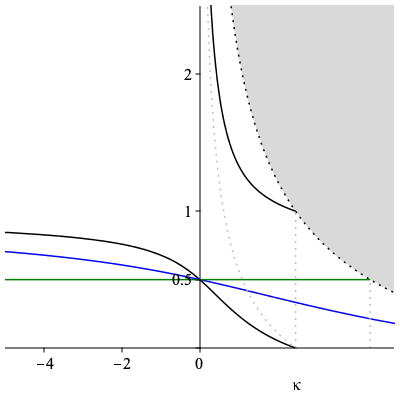}
\put(-80,21){\small $\pi^2/4$}
\put(-33,21){\small $4\pi^2/9$}
\put(-312,73){\small $C2$, Eq. \eqref{eq:GC2}}
\put(-312,92){\small $C1$, Eq. \eqref{eq:GC1} }
\put(-312,106){\small $C3$, Eq. \eqref{eq:GC3}}
\put(-146,248){\small $r_2$}
\put(-109,275){\rotatebox{250}{\rlap{\makebox[2cm]{||}}}}
\put(-72,200){\rotatebox{250}{\rlap{\makebox[2cm]{|||}}}}
\put(-159,150){\rotatebox{340}{\rlap{\makebox[2cm]{|||}}}}
\put(-136,260){\small $C3$, Eq. \eqref{eq:GC3}}
\put(-85,190){\small $r=1+r_2=\pi/ \sqrt{\kappa}$}
\put(-214,155){\small $r=1+r_2=\pi/2 \sqrt{\kappa}$}
\label{F:comp}
\caption{The value of $r_2$ as a function of $\kappa$ according to Eqs. \eqref{eq:GC1}, \eqref{eq:GC3} and \eqref{eq:GC2} under the assumption that $ 2\mu_1=\mu_2$ and $r_1=1$.
Note that for $\kappa>0$ there are two branches for \eqref{eq:GC3} as described in the text.
The shaded area corresponds to values of $(\kappa,r_2)$ that are forbidden since they
violate the restriction that $r=1+r_2<\pi/ \sqrt{\kappa}$.}
\label{F2}
\end{figure}
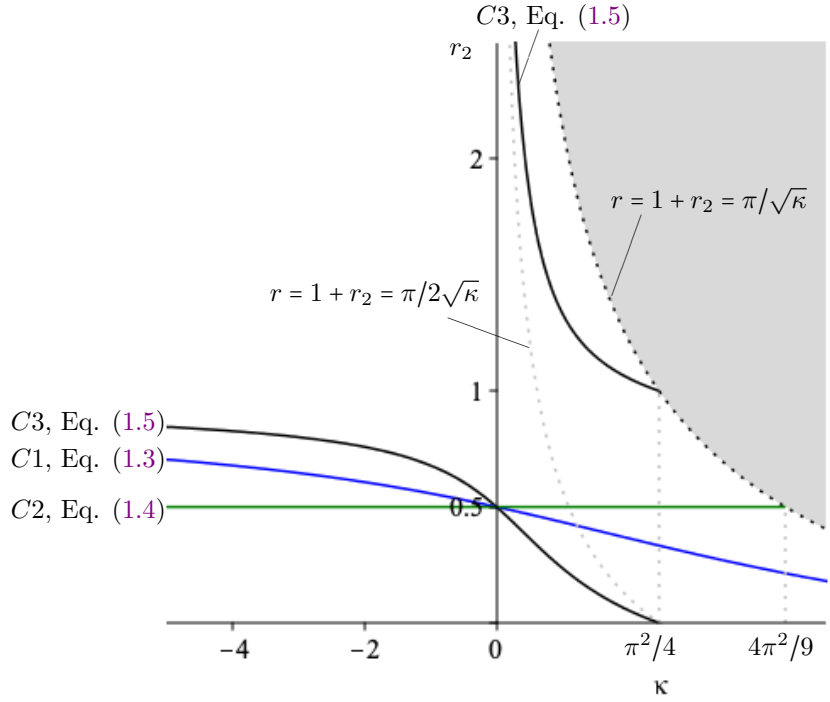

The main body of the paper is devoted to give simple proofs of  Eqs. \eqref{eq:GC1}, \eqref{eq:GC2} and \eqref{eq:GC3}. 
We begin by introducing our models of the spaces of constant curvature in Section \ref{S:basic defns}. 
We then review the construction of Galperin~\cite{Galperin} and extend it to general values of the 
curvature in Section~\ref{S:Galperin} to establish  \eqref{eq:GC1}. Given that the proofs 
of   Eqs.~\eqref{eq:GC2} and~\eqref{eq:GC3} that we present  rely on the 
conservation of momentum,  we devote Section~\ref{S:MomMap} 
to the calculation of the momentum map of the 2-body problem on spaces of constant curvature.
 Once this is done, we give simple  proofs of  \eqref{eq:GC2} and \eqref{eq:GC3}  in
Sections  \ref{S:GC2} and \ref{S:GC3}  respectively. We finish the paper with some conclusions in Section~\ref{S:Conc}.

\begin{remark}
As  indicated  by one of the referees,  the centre of mass $\bar \q$
in  \eqref{eq:c-of-mass} may  also be characterised as the minimiser of the function $F:\R^2\to \R$,
\begin{equation*}
F(\q)=\mu_1 |\q_1-\q|^2+\mu_2 |\q_2-\q|^2.
\end{equation*}
This characterisation can be generalized to  the space $\Sigma_\kappa$ of constant curvature $\kappa$ 
 by requiring that the centre
of mass is the minimiser of $F_\kappa:\Sigma_\kappa \to \R$ given by
\begin{equation*}
F_\kappa(\q)=\mu_1 d_\kappa(\q_1,\q)^2+\mu_2 d_\kappa(\q_2,\q)^2,
\end{equation*}
where $d_\kappa(\q_j,\q)$ is the Riemannian distance between $\q_j$ and $\q$.
Using an approach similar to the one that we follow in Section~\ref{S:Galperin} it
 is not difficult to prove that  this generalisation leads again to the functional relation $\mu_1 r_1=\mu_2 r_2$ 
corresponding to the generalisation of C2.

\end{remark}


\section{Basic working definitions of the spaces of constant curvature}
\label{S:basic defns}

Let $K_\sigma$ denote the diagonal $3\times 3$ matrix $K_\sigma:=\mbox{diag}(1,1,\sigma)$ where $\sigma =\pm1$. The induced bilinear form in $\R^3$ shall be denoted by 
$\langle \cdot , \cdot \rangle_\sigma$, namely,
\begin{equation*}
\langle \bm{u} , \bm{v} \rangle_\sigma :=  \bm{u}^TK_\sigma \vv \in \R, \quad \mbox{for} \quad  \bm{u} , \bm{v}\in \R^3.
\end{equation*}
 Note that $\langle \cdot , \cdot \rangle_{+1}$ is the standard Euclidean scalar product, whereas $\langle \cdot , \cdot \rangle_{-1}$ is the Minkowski pseudo-scalar product. We shall  also denote
\begin{equation*}
\| \vv \|^2_{\sigma}:=\langle \bm{v} , \bm{v} \rangle_\sigma.
\end{equation*}
 Our   model for the (complete and simply connected) space $\Sigma_\kappa$ of non-zero constant Gaussian curvature $\kappa$ is as follows according to the sign of $\kappa$:
\begin{description}
\item[ If $\kappa>0$] then 
\begin{equation*}
\Sigma_\kappa = \left \{ \q \in \R^3 \, : \, \| \q\|^2_{+1}=\frac{1}{\kappa}, \; \mbox{equipped with the restriction of} \; \langle \cdot , \cdot \rangle_{+1}  \right \};
\end{equation*}
e.g. $\Sigma_\kappa$  is the sphere of radius $\frac{1}{\sqrt{\kappa}}$ centred at the origin in $\R^3$, 
equipped with the Riemannian metric that is inherited from the euclidean ambient space. We recall that the 
geodesics on this space are the great circles and that the distance $r\in [0,\pi/\sqrt{\kappa}]$ between two points $\q_1, \q_2\in \Sigma_\kappa$ satisfies
\begin{equation*}
\cos (\sqrt{\kappa} r) = \kappa \langle \q_1, \q_2 \rangle_{+1}.
\end{equation*}

\item[ If $\kappa<0$] then 
\begin{equation*}
\Sigma_\kappa = \left \{ \q = (q_1,q_2,q_3)\in \R^3 \, : \, \| \q\|^2_{-1}=\frac{1}{\kappa}, \;  q_3>0, \; \mbox{equipped with the restriction of} \; \langle \cdot , \cdot \rangle_{-1}  \right \};
\end{equation*}
e.g. $\Sigma_\kappa$  is the upper sheet of the hyperboloid $\| \q\|^2_{-1}=\frac{1}{\kappa}$, which has its vertex at the point $(0,0,1/\sqrt{-\kappa})$, equipped with the Riemannian metric 
which is inherited from the Minkowski pseudo-metric. The geodesics in this case are the hyperbolas obtained as intersections of $\Sigma_\kappa$ with planes passing through 
the origin in $\R^3$, and the distance $r\in [0,\infty)$ between two points $\q_1, \q_2\in \Sigma_\kappa$ satisfies
\begin{equation*}
\cosh (\sqrt{-\kappa} r) = \kappa \langle \q_1, \q_2 \rangle_{-1}.
\end{equation*}

\end{description}

For the rest of the paper it will convenient to note that  $\Sigma_\kappa$ may be parametrised as:
\begin{equation*}
\begin{split}
&\q = \frac{1}{\sqrt{\kappa}}(\cos \varphi \sin \theta, \sin \varphi \sin \theta, \cos \theta), \quad \d s^2= \frac{1}{\kappa} \left ( \d\theta^2+ \sin ^2 \theta\,\d \varphi^2\right ), \qquad \mbox{if $\kappa>0$},  \\
&\q = \frac{1}{\sqrt{-\kappa}}(\cos \varphi \sinh \theta, \sin \varphi \sinh \theta, \cosh \theta), \quad \d s^2= \frac{1}{-\kappa} \left ( \d\theta^2+ \sinh ^2 \theta\,\d \varphi^2\right ), \qquad \mbox{if $\kappa<0$}.
\end{split}
\end{equation*}

\subsection*{Isometries}

We finish this section by recalling that the group $G_\kappa$ of orientation preserving isometries of
$\Sigma_\kappa$ consists of the $3\times 3$ real matrices $g$ with positive determinant and with the property that $g^TK_\sigma g=K_\sigma$, i.e.,
\begin{equation}
\label{eq:G_kappa}
G_\kappa = \begin{cases} \SO(3) \quad \mbox{if} \quad \kappa>0, \\ \SO(2,1) \quad \mbox{if} \quad \kappa<0.\end{cases}
\end{equation}
The action of $G_\kappa$ on $\Sigma_\kappa$ is by standard matrix multiplication.

Regardless of the sign of  $\kappa$, we recall that a fundamental property of  $\Sigma_\kappa$ is that,  
as a Riemannian manifold,
it is both  {\em homogeneous} and {\em isotropic} . Homogeneity
means that for any  $\q_1, \q_2 \in \Sigma_\kappa$ there exists $g\in G_\kappa$ 
such that $g\q_1= \q_2$;   isotropy means that for any $\q\in \Sigma_\kappa$ and any two unit tangent vectors $v_1, v_2 \in T_{\q}\Sigma_\kappa$, there exists $g\in G_\kappa$ such that
$g\q= \q$ and $g_*v_1= v_2$.
  These properties will be used in Sections \ref{S:Galperin}, \ref{S:GC2} and \ref{S:GC3} 
  below to assume, without loss of generality,
   that the masses are located at
 a convenient configuration which simplifies our calculations.

\section{The relativistic  rule of lever \eqref{eq:GC1}.}
\label{S:Galperin}

Consider two masses $\mu_1$, $\mu_2$, located at $\q_1, \q_2\in \Sigma_\kappa$. Following Galperin~\cite{Galperin} we define  their centre of mass $\bar  \q$ as the unique intersection of the ray 
\begin{equation*}
\left \{ s(\mu_1 \q_1 +  \mu_2 \q_2) \, : \, s\in \R \right \}
\end{equation*}
with $\Sigma_\kappa$.  It is shown  in~\cite{Galperin} that this is a well-defined notion that behaves well under the action of isometries and satisfies a set of axioms.

The above definition of centre of mass recovers the standard  lever rule for zero curvature if one realises $\R^2$ as the horizontal plane imbedded   in $\R^3$  by the condition that $q_3=1$. Indeed, if 
$\q_1, \q_2\in \R^2$ then the point
\begin{equation*}
s(\mu_1 (\q_1,1)+\mu_2 (\q_2,1))\in \R^3
\end{equation*}
has third component equal to $1$ if and only if $s=\frac{1}{\mu_1+\mu_2}$. Hence, the above definition of the centre of mass recovers Eq. \eqref{eq:c-of-mass}.
Below we show  that for  $\kappa \neq 0$,  Galperin's definition leads to the relativistic rule of lever~\eqref{eq:GC1}.

%
%
%
%
%
\paragraph{Case $\kappa>0$.} Because of the
homogeneity and isotropy of $\Sigma_\kappa$ we may suppose  without loss of generality that the 
masses are located at 
\begin{equation*}
\q_1 =\frac{1}{\sqrt{\kappa}}(0,-\sin \alpha_1, \cos \alpha_1 ), \quad \q_2=\frac{1}{\sqrt{\kappa}}(0,\sin \alpha_2, \cos \alpha_2 ) \in \Sigma_\kappa, \qquad \alpha_1, \, \alpha_2\in (0,\pi/2),
\end{equation*}
and that, according to Galperin's definition,  their centre of mass   is the north pole $\frac{1}{\sqrt{\kappa}}(0,0,1 )\in \Sigma_\kappa$. The condition that 
\begin{equation*}
s(\mu_1 \q_1 +  \mu_2 \q_2)=\frac{1}{\sqrt{\kappa}}(0,0,1 ),
\end{equation*}
is satisfied for an $s\in \R$ if and only if $\mu_1 \sin \alpha_1 =\mu_2\sin \alpha_2$. Considering that the Riemannian distance from $\q_j$ to the north pole is $r_j=\alpha_j/\sqrt{\kappa}$, $j=1,2$, we 
obtain  $\mu_1 \sin (\sqrt{\kappa} r_1) =\mu_2 \sin (\sqrt{\kappa} r_2)$, as required.

\paragraph{Case $\kappa<0$.} The proof is analogous to the above. This time, owing to the
homogeneity and isotropy of $\Sigma_\kappa$ we may suppose   without loss of generality  that the 
masses are located at 
\begin{equation*}
\q_1 =\frac{1}{\sqrt{-\kappa}}(0,-\sinh \alpha_1, \cosh \alpha_1 ), \quad \q_2=\frac{1}{\sqrt{-\kappa}}(0,\sinh \alpha_2, \cosh \alpha_2 ) \in \Sigma_\kappa, \qquad \alpha_1, \, \alpha_2\in (0,\infty),
\end{equation*}
and that, according to Galperin's definition,  their centre of mass   is the hyperboloid's vertex $\frac{1}{\sqrt{-\kappa}}(0,0,1 )\in \Sigma_\kappa$. As before, the condition
that 
\begin{equation*}
s(\mu_1 \q_1 +  \mu_2 \q_2)=\frac{1}{\sqrt{-\kappa}}(0,0,1 ),
\end{equation*}
is satisfied for an $s\in \R$ if and only if $\mu_1 \sinh \alpha_1 =\mu_2\sinh \alpha_2$. Given that the Riemannian distance from $\q_j$ to the  hyperboloid's vertex 
 is $r_j=\alpha_j/\sqrt{-\kappa}$, $j=1,2$, we 
obtain  $\mu_1 \sinh (\sqrt{-\kappa} r_1) =\mu_2 \sinh (\sqrt{-\kappa} r_2)$, as required.

\section{The conserved momentum of the 2-body problem on spaces of constant Gaussian curvature}
\label{S:MomMap}

In sections \ref{S:GC2} and \ref{S:GC3}  ahead we give proofs of   Eqs. \eqref{eq:GC2} and \eqref{eq:GC3}. Such proofs rely entirely on the conservation of momentum. In the zero-curvature case one may prove that  \eqref{eq:levRule0} arises a consequence of C2 and C3 by using  the  conservation of the linear momentum:
\begin{equation}
\label{eq:linmom}
\bm{p}=\mu_1 \dot \q_1+\mu_2 \dot \q_2.
\end{equation}
A similar proof may be given for $\kappa\neq 0$ but one requires the full components of the momentum map.  The purpose of this section is to compute this momentum map which is given in Proposition~\ref{prop:momentum} below. 

\paragraph{The configuration space and the Lagrangian.} The configuration space for the 2-body problem in $\Sigma_\kappa$ is
\begin{equation*}
M_\kappa= \Sigma_\kappa\times \Sigma_\kappa \setminus \Delta_\kappa,
\end{equation*}
where $\Delta_\kappa$ denotes the set of collision configurations if $\kappa<0$; and the set of 
collision and antipodal configurations if  $\kappa>0$.
The Lagrangian $L:TM_\kappa\to \R$ is given by
\begin{equation}
\label{eq:Lag}
L(\q_1,\q_2, \dot \q_1,\dot \q_2)= \frac{\mu_1}{2} \| \dot \q_1 \|_\sigma^2+ \frac{\mu_2}{2} \| \dot \q_2 \|_\sigma^2-V_\kappa(r), \qquad \sigma:=\mbox{sign}(\kappa),
\end{equation}
where, as usual, $\mu_1, \mu_2>0$ are the particles' masses, 
$r>0$ is the Riemannian distance between $\q_1,\q_2\in \Sigma_\kappa$, the velocity vectors $ \dot \q_j \in T_{\q_j}\Sigma_\kappa$, $j=1,2$, and
$V_\kappa:I_\kappa \to \R$ is an attractive potential e.g. $V_\kappa'(r)>0$. The domain $I_\kappa$ of 
$V_\kappa$ varies with $\kappa$: it is the infinite interval  $(0,\infty)$ if $\kappa<0$ and the 
finite interval $(0,\pi/\sqrt{\kappa})$ if $\kappa>0$.\footnote{ The accepted generalisation of 
Newton's $1/r$ gravitational  law to spaces of constant curvature requires that
 $V_\kappa$ is proportional to $- \sqrt{\kappa}\cot (\sqrt{\kappa} r)$ if $\kappa>0$ and to $- \sqrt{-\kappa}\coth (\sqrt{-\kappa} r)$ if $\kappa<0$ (see e.g. \cite{Kozlov-Harin, Carinena}). The results of this paper are valid for more general attractive potentials 
$V_\kappa$.}

\paragraph{Symmetries.} The group $G_\kappa$ of orientation preserving isometries of $\Sigma_\kappa$ (given by \eqref{eq:G_kappa} above) acts diagonally on $M_\kappa$, i.e.
\begin{equation*}
g\cdot (\q_1,\q_2) = (g\q_1,g\q_2).
\end{equation*}
and its tangent lift leaves the Lagrangian
\eqref{eq:Lag} invariant. 

The corresponding Lie algebra $\g_\kappa$ is formed by the $3\times 3$ real matrices $\xi$ that satisfy $\xi^TK_\sigma  + K_\sigma \xi=0$, i.e.:
\begin{equation*}
 \mathfrak{g}_\kappa = \begin{cases} \so(3) \quad \mbox{if $\kappa>0$}, \\ \so(2,1) \quad \mbox{if $\kappa<0$}. \end{cases}
\end{equation*}

\paragraph{Momentum map.} According to the general theory of lifted actions for mechanical systems
\cite{MaRa}  there exists a momentum map  $\bm{J}:TM_\kappa\to \mathfrak{g}_\kappa^*$, that  is 
conserved along the solutions of the equations of motion defined by the Lagrangian $\eqref{eq:Lag}$. Considering that the infinitesimal action of $\xi\in \g_\kappa$ on $(\q_1,\q_2)\in M_\kappa$ is again linear,
i.e.
\begin{equation*}
\xi \cdot (\q_1,\q_2) = (\xi \q_1,\xi\q_2),
\end{equation*}
the momentum map $\bm{J}$ is defined by:
\begin{equation}
\label{eq:MomMapAbstract} 
\left \langle \bm{J}(\q_1,\q_2, \dot \q_1,\dot \q_2) , \xi \right \rangle_{\g_\kappa} = \mu_1 \left \langle \xi \q_1, \dot \q_1 \right  \rangle_{\sigma}+  \mu_2 \left \langle \xi \q_2, \dot \q_2 \right  \rangle_{\sigma},
\end{equation}
where $\langle \cdot , \cdot  \rangle_{\g_\kappa}$ denotes the dual pairing between $\g_\kappa^*$ and  $\g_\kappa$.

In order to obtain an explicit expression for $\bm{J}(\q_1,\q_2, \dot \q_1,\dot \q_2)$,  we identify $\g_\kappa$ with $\R^3$ as vector spaces by introducing the following ordered basis
of $\g_\kappa$:
\begin{equation*}
\xi_1=\begin{pmatrix} 0 & 0 & 0 \\ 0& 0 & - \sigma \\ 0 &  1 & 0 \end{pmatrix}, \qquad \xi_2=\begin{pmatrix} 0 & 0 & \sigma \\ 0& 0 & 0 \\ -1 &  0 & 0 \end{pmatrix}, \qquad   
 \xi_3=\begin{pmatrix} 0 & -1 & 0 \\ 1 & 0 & 0 \\ 0 &  0 & 0 \end{pmatrix}.
\end{equation*}
We also  identify the dual space of $\R^3$ with itself via the euclidean scalar product. Under these identifications, the range of the momentum map is $\R^3$ and we have:
\begin{proposition}
\label{prop:momentum}
The momentum map 
\begin{equation}
\label{eq:MomMap} 
\bm{J}:TM_\kappa \to \R^3, \qquad ( \q_1,\q_2, \dot \q_1,\dot \q_2) \mapsto \mu_1 (K_\sigma \q_1)\times (K_\sigma \dot \q_1)+ \mu_2 (K_\sigma \q_2)\times (K_\sigma \dot \q_2),
\end{equation}
where $\times$ denotes the standard vector product in $\R^3$ and, as before, $\sigma=\mbox{\em sign}(\kappa)$.
\end{proposition}
\begin{proof}
We only consider the case $\kappa<0$, i.e. $\sigma=-1$; the other case is simpler and quite standard. Introduce the notation
\begin{equation*}
\begin{split}
&\q_j=(x_j,y_j,z_j), \quad \dot \q_j=(\dot x_j,\dot y_j,\dot z_j), \quad j=1,2, \\
& \xi= (a_1,a_2,a_3)\in \R^3, \; \mbox{i.e. $\xi$ represents the matrix $\sum_{j=1}^3a_j\xi_j\in \so(2,1)$.}
\end{split}
\end{equation*}
A direct calculation shows that the right hand side of \eqref{eq:MomMapAbstract} equals:
\begin{equation*}
\sum_{j=1}^2\mu_j\left ( (-a_2z_j-a_3y_j)\dot x_j + (a_1z_j + a_3 x_j) \dot y_j + (-a_1y_j+a_2x_j)\dot z_j\right ),
\end{equation*}
which may be rewritten as the euclidean scalar product of $\xi=(a_1,a_2,a_3)\in \R^3$ with the vector
\begin{equation*}
\mu_1 (K_{-1} \q_1)\times (K_{-1} \dot \q_1)+\mu_2 (K_{-1} \q_2)\times (K_{-1} \dot \q_2) .
\end{equation*}

\end{proof}

\section{Proof of Eq. \eqref{eq:GC2}}
\label{S:GC2}

Assume that the particles are under the influence of an attractive potential depending only on their mutual distance.
We prove that the characterisation of the centre of mass as the point of collision of two particles which are initially at rest leads to the relation $\mu_1 r_1 =\mu_2r_2$ 
independently of the value of the curvature $\kappa$ and of the specific form of the attractive potential.
For the sake of completeness we begin by proving that such  formula holds in the usual case
$\kappa=0$. 

\paragraph{Case $\kappa=0$.} We may suppose, without loss of generality, that the masses at rest are located at 
$t=0$ at
\begin{equation*}
\q_1(0)=(-r_1,0), \quad \q_2(0)=(r_2,0) \in \R^2, \qquad r_1, \, r_2>0,
\end{equation*}
and  that at time $T>0$  they collide at the origin of $\R^2$. 
Due to the nature of the attractive force, 
the trajectories of the particles are given by
\begin{equation}
\label{eq:trajcollision}
\q_1(t)=(x_1(t),0), \quad \q_2(t)=(x_2(t),0),
\end{equation}
where $x_1(0)=-r_1$, $x_2(0)=r_2$, and  $x_1(T)=x_2(T)=0$.  Considering that the value of the linear momentum at $t=0$ is zero, the conservation of~\eqref{eq:linmom} yields:
\begin{equation*}
\mu_1 \dot x_1(t) + \mu_2 \dot x_2(t)=0,
\end{equation*}
which, after integration  from $0$ to $T$ gives~$\mu_1 r_1 =\mu_2r_2$  as required.

\begin{remark}
A  proof that the second component of $\q_j(t)$,  $j=1,2$, in~\eqref{eq:trajcollision} is identically zero 
may be given using  the symmetry argument reproduced  in the souvenir coffee mug of the
conference in  honour of James Montaldi in Guanajuato in 2018:  
 the problem is   equivariant  under reflections about the $x$-axis and the initial condition is fixed by this reflection. A similar reflection argument
also proves that the first component of $\q_j(t)$,  $j=1,2$, in the trajectories \eqref{eq:trajcollision-sphere} and \eqref{eq:trajcollision-pseudosphere}  below indeed vanishes. 
\end{remark}

\paragraph{Case $\kappa>0$.} The proof is completely analogous to the above. Because of the
homogeneity and isotropy of $\Sigma_\kappa$ we may suppose that the 
particles are initially located at 
\begin{equation*}
\q_1(0)=\frac{1}{\sqrt{\kappa}}(0,-\sin \alpha_1, \cos \alpha_1 ), \quad \q_2(0)=\frac{1}{\sqrt{\kappa}}(0,\sin \alpha_2, \cos \alpha_2 ) \in \Sigma_\kappa, \qquad \alpha_1, \, \alpha_2\in (0,\pi/2),
\end{equation*}
and they collide at time $T>0$ at the north pole $\frac{1}{\sqrt{\kappa}}(0,0, 1 )\in \Sigma_\kappa$. The particles then follow trajectories 
\begin{equation}
\label{eq:trajcollision-sphere}
\q_1(t)=\frac{1}{\sqrt{\kappa}}(0,-\sin \theta_1(t), \cos \theta_1(t) ), \quad \q_2(t)=\frac{1}{\sqrt{\kappa}}(0,\sin \theta_2(t), \cos \theta_2(t) ), 
\end{equation}
where $\theta_j(0)=\alpha_j$, $\theta_j(T)=0$, $j=1,2$.
The momentum~\eqref{eq:MomMap} along these trajectories is computed to be:
\begin{equation*}
\bm{J}(\q_1(t),\q_2(t), \dot \q_1(t), \dot \q_2(t))= \frac{1}{\kappa}(\mu_1\dot \theta_1(t) -\mu_2\dot \theta_2(t),0,0), 
\end{equation*}
and should be identically zero to agree with the value of $\bm{J}$ at time $t=0$. Integrating the equation $\mu_1\dot \theta_1(t) -\mu_2\dot \theta_2(t)=0$ from $0$ to $T$ leads to
$\mu_1\alpha_1=\mu_2\alpha_2$. The proof that $\mu_1 r_1 =\mu_2r_2$ is completed by noting that the Riemannian distance from $\q_j(0)$ to the north pole is $r_j=\alpha_j/\sqrt{\kappa}$, $j=1,2$.

\paragraph{Case $\kappa<0$.} The proof is again analogous to the cases $\kappa=0$ and $\kappa>0$. This time we assume without loss of generality that the initial position of the particles is
\begin{equation*}
\q_1(0)=\frac{1}{\sqrt{-\kappa}}(0,-\sinh \alpha_1, \cosh \alpha_1 ), \quad \q_2(0)=\frac{1}{\sqrt{-\kappa}}(0,\sinh \alpha_2, \cosh \alpha_2 ) \in \Sigma_\kappa, 
\end{equation*}
$\alpha_1, \, \alpha_2\in (0,\infty)$, and that they collide at time $T>0$ at the hyperboloid's vertex  $\frac{1}{\sqrt{-\kappa}}(0,0, 1 )\in \Sigma_\kappa$. The  trajectories of the particles are contained on the geodesic passing
through $\q_1(0)$ and $\q_2(0)$  and
are given by
\begin{equation}
\label{eq:trajcollision-pseudosphere}
\q_1(t)=\frac{1}{\sqrt{-\kappa}}(0,-\sinh \theta_1(t), \cosh \theta_1(t) ), \quad \q_2(t)=\frac{1}{\sqrt{-\kappa}}(0,\sinh \theta_2(t), \cosh \theta_2(t) ), 
\end{equation}
where  $\theta_j(0)=\alpha_j$, $\theta_j(T)=0$, $j=1,2$.
The momentum~\eqref{eq:MomMap} along these motions simplifies to:
\begin{equation*}
\bm{J}(\q_1(t),\q_2(t), \dot \q_1(t), \dot \q_2(t))= \frac{1}{\kappa}(\mu_1\dot \theta_1(t) -\mu_2\dot \theta_2(t),0,0),
\end{equation*}
and once again should vanish identically to agree with the value of $\bm{J}$ at time $t=0$. Integrating the equation $\mu_1\dot \theta_1(t) - \mu_2\dot \theta_2(t)=0$ from $0$ to $T$ leads to
$\mu_1\alpha_1=\mu_2\alpha_2$. As before, the proof that $\mu_1 r_1 =\mu_2r_2$ is completed by noting that the Riemannian distance from $\q_j(0)$ to the 
hyperboloid's vertex is $r_j=\alpha_j/\sqrt{-\kappa}$, $j=1,2$.

\section{Proof   of Eq. \eqref{eq:GC3}}
\label{S:GC3}

As in the previous section, we assume that the  particles are under the influence of an attractive potential depending only on their mutual distance.
We prove that, independently of the form of the attractive force,  the characterisation of the centre of mass as the centre of rotation of  uniformly rotating solutions that preserve the distance between the particles leads to the
following relations depending on the curvature $\kappa$:
\begin{equation}
\label{eq:CentreofSteadyRotation}
\begin{cases}
\mu_1r_1=\mu_2r_2 \quad \mbox{if $\kappa=0$;} \\
\mu_1 \sin ( 2\sqrt{\kappa} r_1) =\mu_2 \sin ( 2\sqrt{\kappa} r_2)\quad \mbox{if $\kappa>0$;}  \\
\mu_1 \sinh ( 2\sqrt{-\kappa} r_1) =\mu_2 \sinh ( 2\sqrt{-\kappa} r_2)\quad \mbox{if $\kappa<0$.} 
\end{cases}
\end{equation}

\begin{remark}
To be precise, in  this section we only prove that \eqref{eq:CentreofSteadyRotation} are  necessary conditions for the
existence of uniformly rotating solutions in which the distance between the particles remains constant. The existence of 
this kind of solutions when $\kappa\neq 0$ for arbitrary attractive potentials was recently proved in \cite{BGNMM}.
\end{remark}

Throughout this section we consider  $\q_1, \q_2$,  $\bm{p}$ and the image of $\bm{J}$  as column vectors.

\paragraph{Case $\kappa=0$.}  As in the  previous section, we give a proof of the case $\kappa=0$ for completeness. We again suppose,
without loss of generality, that 
at  time $t=0$ the masses lie at
\begin{equation*}
\q_1(0)=\begin{pmatrix} -r_1 \\ 0 \end{pmatrix}, \quad \q_2(0)=\begin{pmatrix} r_2 \\ 0 \end{pmatrix} \in \R^2, \qquad r_1, \, r_2>0,
\end{equation*}
but this time we suppose they rotate uniformly about the origin maintaining a constant distance between them at all time.
 According to these assumptions, the particles follow trajectories
\begin{equation*}
\q_j(t)=\begin{pmatrix} \cos \omega t & -\sin \omega t \\ \sin \omega t & \cos \omega t \end{pmatrix}\q_j(0), \qquad j=1,2,
\end{equation*}
for a certain angular speed $0\neq \omega\in \R$. Therefore, their linear momentum~\eqref{eq:linmom} equals
 \begin{equation*}
\bm{p}=\omega (\mu_1r_1-\mu_2r_2)\begin{pmatrix}  \sin \omega t \\  -\cos \omega t \end{pmatrix},
\end{equation*}
which is constant if and only if~$\mu_1r_1=\mu_2r_2$.

\paragraph{Case $\kappa>0$.} Suppose without loss of generality
 that the particles are initially positioned at
\begin{equation*}
\q_1(0)=\frac{1}{\sqrt{\kappa}} \begin{pmatrix} 0 \\ -\sin \alpha_1 \\ \cos \alpha_1 \end{pmatrix},
 \quad 
\q_2(0)=\frac{1}{\sqrt{\kappa}} \begin{pmatrix} 0 \\ \sin \alpha_2 \\  \cos \alpha_2 \end{pmatrix} \in \Sigma_\kappa, \qquad \alpha_1, \, \alpha_2\in (0,\pi/2),
\end{equation*}
and that  they rotate uniformly about the north pole  maintaining a constant distance between them at all time. The particles then follow trajectories
\begin{equation*}
\begin{split}
\q_j(t)&=
\begin{pmatrix}
\cos \omega t & -\sin \omega t & 0 \\ \sin \omega t & \cos \omega t  & 0 \\ 0 & 0 &1 \end{pmatrix}\q_j(0), \quad j=1,2,
\end{split}
\end{equation*}
for an angular speed $0\neq \omega\in \R$. A direct calculation shows that their momentum~\eqref{eq:MomMap} equals:
\begin{equation*}
\bm{J}(\q_1(t),\q_2(t), \dot \q_1(t), \dot \q_2(t))=  \frac{\omega}{2\kappa} \begin{pmatrix} - (\mu_1\sin 2\alpha_1 - \mu_2 \sin 2\alpha_2)\sin \omega t \\  (\mu_1\sin 2\alpha_1 - \mu_2 \sin 2\alpha_2)\cos \omega t \\2(\mu_1\sin^2\alpha_1 + \mu_2\sin^2\alpha_2) \end{pmatrix},
\end{equation*}
which is constant if and only if $\mu_1\sin 2\alpha_1 = \mu_2 \sin 2\alpha_2$. Given that the Riemannian distance from $\q_j(0)$ to the north pole is $r_j=\alpha_j/\sqrt{\kappa}$, $j=1,2$, we conclude  that the momentum is constant if and only if $\mu_1\sin (2\sqrt{\kappa} r_1) = \mu_2 \sin  (2\sqrt{\kappa} r_2)$ as required.

\paragraph{Case $\kappa<0$.} Finally, we suppose (without loss of generality) 
that the initial position of the particles is
\begin{equation}
\label{eq:initial-hyp}
\q_1(0)=\frac{1}{\sqrt{-\kappa}} \begin{pmatrix} 0 \\ -\sinh \alpha_1 \\ \cosh \alpha_1 \end{pmatrix}, \quad \q_2(0)=\frac{1}{\sqrt{-\kappa}} \begin{pmatrix} 0\\ \sinh \alpha_2 \\
 \cosh \alpha_2 \end{pmatrix} \in \Sigma_\kappa,
%
\end{equation}
$\alpha_1, \, \alpha_2\in (0,\infty)$, and that they rotate uniformly around the hyperboloid's vertex  $\frac{1}{\sqrt{-\kappa}}(0,0, 1 )\in \Sigma_\kappa$ maintaining a constant distance between them at all time. The trajectory followed by the
particles is now given by
\begin{equation*}
\begin{split}
\q_j(t)&=\begin{pmatrix}
\cos \omega t & -\sin \omega t & 0 \\ \sin \omega t & \cos \omega t  & 0 \\ 0 & 0 &1 \end{pmatrix}\q_j(0), \quad j=1,2,
\end{split}
\end{equation*}
for an angular speed $0\neq \omega\in \R$. The momentum~\eqref{eq:MomMap} along these trajectories equals:
\begin{equation*}
\bm{J}(\q_1(t),\q_2(t), \dot \q_1(t), \dot \q_2(t))=  \frac{\omega}{2\kappa} \begin{pmatrix} - (\mu_1\sinh 2\alpha_1 - \mu_2 \sinh 2\alpha_2)\sin \omega t \\  (\mu_1\sinh 2\alpha_1 - \mu_2 \sinh 2\alpha_2)\cos \omega t \\ -2(\mu_1\sinh^2\alpha_1 + \mu_2\sinh^2\alpha_2) \end{pmatrix},
\end{equation*}
which is constant if and only if $\mu_1\sinh 2\alpha_1 = \mu_2 \sinh 2\alpha_2$. Given that the Riemannian distance from $\q_j(0)$ to the vertex  $\frac{1}{\sqrt{-\kappa}}(0,0, 1 $) is $r_j=\alpha_j/\sqrt{-\kappa}$, $j=1,2$, we conclude  that the momentum is constant if and only if $\mu_1\sinh (2\sqrt{-\kappa} r_1) = \mu_2 \sinh  (2\sqrt{-\kappa} r_2)$ as required.

\begin{remark}
For $\kappa<0$ there exists a different kind of stationary motion  with the property that the distance between the particles 
remains constant for all $t$. These are the so-called {\em hyperbolic relative equilibria}~\cite{DPR,GNMPR,BGNMM} which are unbounded solutions that, for the initial condition~\eqref{eq:initial-hyp}, correspond to a `hyperbolic rotation'
\begin{equation}
\label{eq:hyp-rot}
\begin{split}
\q_j(t)=\begin{pmatrix}
\cosh \omega t & 0 & \sinh \omega t  \\0  & 1  & 0 \\ \sinh \omega t & 0 & \cosh \omega t \end{pmatrix}\q_j(0), \qquad j=1,2,
\end{split}
\end{equation}
for a certain `rotation speed' $0\neq \omega\in \R$. These solutions exist as a balance of the
gravitational force and the tendency of the geodesics to `spread out' when the curvature is negative
(see the discussion in \cite{GNMPR}). 
Along such solutions one cannot talk of a fixed centre of rotation. However, it is interesting to note
that the moving point
\begin{equation*}
\hat \q (t) := \frac{1}{\sqrt{-\kappa}} \begin{pmatrix} \sinh \omega t \\ 0 \\ \cosh \omega t \end{pmatrix}
\end{equation*}
traverses a geodesic at constant speed 
and maintains a constant distance $r_j=\alpha_j/\sqrt{-\kappa}$,  with $\mu_j$,
$j=1,2$. This property is reminiscent of the  solutions of the two-body problem where the center of mass travels at constant non-zero speed and
the particles rotate uniformly about it.

Interestingly, the condition that $\mu_1\sinh (2\sqrt{-\kappa} r_1) = \mu_2 \sinh  (2\sqrt{-\kappa} r_2)$ is also necessary 
for the existence of this type of solutions. Indeed,
the momentum~\eqref{eq:MomMap} along the trajectory~\eqref{eq:hyp-rot} equals:
\begin{equation*}
\bm{J}(\q_1(t),\q_2(t), \dot \q_1(t), \dot \q_2(t))=  \frac{\omega}{2\kappa} \begin{pmatrix} - (\mu_1\sinh 2\alpha_1 - \mu_2 \sinh 2\alpha_2)\sinh \omega t \\  2(\mu_1\cosh^2\alpha_1 + \mu_2\cosh^2\alpha_2) \\ -(\mu_1\sinh 2\alpha_1 - \mu_2 \sinh 2\alpha_2)\cosh \omega t   \end{pmatrix},
\end{equation*}
which is constant if and only if $\mu_1\sinh 2\alpha_1 = \mu_2 \sinh 2\alpha_2$.
\end{remark}

\section{Final remarks}
\label{S:Conc}
We have given evidence to  show that the generalisation of the notion of centre of mass to spaces of
 non-zero constant curvature is
not straightforward. In particular, we have shown that (for distinct masses)
 the relativistic lever rule  proposed by Galperin \cite{Galperin}
does not possess some basic dynamical properties of the centre of mass of the classical 2-body problem in $\R^2$.
A natural question is whether there is a sensible definition of the centre of mass that is relevant for the 
analysis of the 2-body problem
 in surfaces of non-zero constant curvature. Below I explain why  this paper shows that the   
answer to this question is negative. This conclusion is in agreement with observations made before, e.g. \cite{Diacu}, and
seems to be related with the absence of Galilean boosts for the problem in the case of non-zero curvature.

The fundamental property of the centre of mass of the classical $2$-body problem in $\R^2$ is that  it   travels at 
constant velocity along {\em all}
solutions of the problem. Hence, the question is if one can define a centre of mass in a space of constant non-zero curvature 
(solely in terms of the masses and positions of the particles) with the property that it travels along 
a geodesic at  a constant speed along {\em all}
solutions of the problem. In this paper we have considered the collision and steady rotation solutions, which are perhaps the
simplest solutions  to the 
2-body problem. The collision point and the centre of steady rotation indeed satisfy the requirement of travelling along a geodesic
at constant speed (equal to zero). The fact that these points are determined by  distinct relations - \eqref{eq:GC2} and
\eqref{eq:GC3} - contradicts the existence of the  definition of centre of mass with the desired properties.

\section*{Acknowledgements}
I am thankful to James Montaldi for several conversations during the last years on many versions
of the  2-body problem. I acknowledge support for my research from the Program 
UNAM-DGAPA-PAPIIT-IN115820  and from the
Alexander von Humboldt Foundation for a Georg Forster Experienced Researcher Fellowship that funded a research visit
to TU Berlin where part of this work was done. Finally, I wish to acknowledge the referees for their useful comments
 that helped me to improve this paper.
\vskip 1cm

%
%

\begin{thebibliography}{99}
\let\\, \newcommand{\by}[1]{\textsc{\ignorespaces #1}\\}
  \newcommand{\title}[1]{\textsl{\ignorespaces #1}\\}
  \newcommand{\vol}[1]{{\bf{\ignorespaces #1}}}
  \newcommand{\info}[1]{\textrm{\ignorespaces #1}.}

\small  \setlength{\parskip}{0pt}
  
%
%
%
%
%
%
%
%
%
%
%
%
%
%
%
%
%
%
%
%
%
%
%
%
%

%

\bibitem{BMK} \by{Borisov ~A.V., Mamaev~I.S. and A.A.~Kilin} Two-body problem on a sphere. Reduction, stochasticity,  periodic orbits. {\em  Regul. Chaotic Dyn.} \vol{9}, (2004) 265--279. 


\bibitem{BGNMM} \by{ Borisov ~A.V.,  Garc\'ia-Naranjo~L.C., Mamaev~I.S. and  J.~Montaldi} Reduction and relative equilibria for the two-body problem on spaces of constant curvature. 
{\em Celest. Mech. Dyn. Astr.} \vol{130} (2018), 36pp.

\bibitem{Carinena}
	\by{Cari\~nena~J.F., Ra\~nada~M.F. and  M.~Santander}
	Central potentials on spaces of constant curvature: the Kepler problem 
	on the two-dimensional sphere $S^2$ and the hyperbolic plane $H^2$. 
	\textit{J.\ Math.\ Phys.} \textbf{46} (2005), 052702. 


 \bibitem{Diacu}  \by{Diacu~F.}, The non-existence of centre of mass and linear momentum integrals in the curved N-body problem,
{\em Libertas Math.} \vol{32} (2012) 25--37.

\bibitem{DPR} \by{Diacu~F., P\'erez-Chavela~E and J.G.~Reyes} An intrinsic approach in the curved $n$-body problem. The negative case,
{\em J. Differential Equations} \vol{252} (2012) 4529--4562.


%
%
%
%
%
%
%
%
%
%
%
%
%
%
%
%
%
%
%
%
%
%
%
%
%
%
%
%
%
%
%
%
%
%
%
%
%
%

%
%
%

\bibitem{Galperin} \by{Galperin G.A.} A concept of the mass center of a system of material points in the constant curvature spaces, {\em Comm.
Math. Phys.} {\bf 154 } (1993) 63--84.

\bibitem{GNMPR}\by {Garc\'ia-Naranjo L.~C.,  Marrero J.~C.,  P\'erez-Chavela, E. and M. Rodr\'iguez-Olmos} Classification and stability of
relative equilibria for the two-body problem in the hyperbolic space of dimension 2. {\em J. Differ. Equ.}  {\bf 260} (2016)
6375--6404.

\bibitem{GNM}\by {Garc\'ia-Naranjo L.~C. and J. Montaldi} Attracting and repelling 2-body problems on a family of surfaces of constant curvature. arXiv:1906.01070.

%
%
%


%
%
%
%
%
%
%
%
%
%
%
%

\bibitem{Kozlov-Harin} \by{Kozlov~V.V. and  A.O.~Harin} Kepler's problem in constant curvature spaces, 
\emph{Celestial Mech. Dynam. Astronom.} \vol{54}
(1992) 393--399.


%
%
%
%
%
%
%
%
\bibitem{MaRa} \by{Marsden J.E. and T.S.~Ratiu} \title{Introduction to Mechanics with Symmetry} 
\info{Texts in Applied Mathematics  \vol{17} Springer-Verlag 1994}


%
%
%
%
%
%
%
%
%
%

%
%
%
%
%
%
%


\end{thebibliography}
\end{document}